\documentclass[superscriptaddress, reprint, , amssymb, aps, prx, floatfix]{revtex4-2}

\usepackage[utf8]{inputenc}
\usepackage{amsmath, amsthm, amssymb}
\usepackage{mathtools}
\usepackage[dvipsnames]{xcolor}
\usepackage{bm}
\usepackage{amsfonts}
\usepackage{braket}
\usepackage{graphicx}
\usepackage{caption}
\usepackage{subcaption}
\usepackage{color} 
\usepackage{hyperref}
\usepackage{url}
\usepackage{units}
\usepackage{comment}
\usepackage{booktabs}
\usepackage{pgf}

\newtheorem{theorem}{Theorem}

\newtheorem{proposition}[theorem]{Proposition}

\usepackage{tikz}
\usetikzlibrary{calc}
\usepackage{pgfplots}
\pgfplotsset{compat=1.3}

\hypersetup{colorlinks = true, pdfauthor = , pdftitle = A QOC perspective on VQAs}

\newcommand{\needcite}[1]{\textcolor{red}{[Ref needed]}}

\newcommand{\B}{\mathbb{B}}
\newcommand{\R}{\mathbb{R}}

\newcommand{\N}{\mathbb{N}}
\newcommand{\T}{\intercal}

\usepackage{qcircuit}

\graphicspath{ {./images/} }

\begin{document}

\title{Power network optimization: a quantum approach}

\author{Giuseppe Colucci}
\affiliation{Quantum Quants, Rotterdam, The Netherlands}

\author{Stan van der Linde}
\affiliation{TNO, The Hague, The Netherlands}

\author{Frank Phillipson}
\affiliation{TNO, The Hague, The Netherlands}
\affiliation{Maastricht University, School of Business and Economics, The Netherlands}

\date{\today}

\begin{abstract}

Optimization of electricity surplus is a crucial element for transmission power networks to reduce costs and efficiently use the available electricity across the network. In this paper we showed how to optimize such a network with quantum annealing. First, we define the QUBO problem for the partitioning of the network, and test the implementation on purely quantum and hybrid architectures. We then solve the problem on the D-Wave hybrid CQM and BQM solvers, as well as on classical solvers available on Azure Quantum cloud. Finally, we show that the hybrid approaches overperform the classical methods in terms of quality of the solution, as the value of the objective function of the quantum solutions is found to be always lower than with the classical approaches across a set of different problem size. 
\end{abstract}

\maketitle
\onecolumngrid
\section{Introduction}\label{sec:introduction}
Across many different fields, optimization is a crucial element used for solving problems.
From scientific to industrial applications, optimization is used to generally boost performance and reduce risks.

In business operations, optimization models play an important role in decisions making, from a financial (e.g., decision on investments) as well as from a non-financial perspective (e.g. assessment of regulatory risks or decrease of ecological footprint). 

In the recent years, the energy sector has seen a global shift toward clean energy with wind and solar power, where companies accelerate ambitious renewable energy goals around the world. Therefore, the energy market is gradually shifting from centrally planned, financed and operated electricity systems to a more diffuse, increasingly decentralized and real-time model where the planning, finance and operation of the system is shared between an increasing number of parties. In fact, in a decentralized ecosystem new innovations enter the utility space, from prosumers and photovoltaic panels (PVs) to batteries and electric vehicles (EVs).

Decentralization can bring benefits but it also increases the complexity of electricity systems. In fact, on one hand it allows for increasingly better use of renewable energy sources as well as combined heat and power, reduces fossil fuel use and increases ecoefficiency \cite{soininen2021law}. On the other hand, the rising number of energy players and entanglement of the underlying energy network increases the challenges in optimizing costs and operations of the entire network.

Several methods have been implemented to improve the efficiency as well as the resiliency of such networks \cite{Shahrokhi1990TheMC,Edmonds2001TheoreticalII,Bai2008SemidefinitePF,Nguyen12,Yamagata2015ProposalFA,Tanjo2016GraphPO,Ghaddar2016OptimalPF,Bella2020SupervisedMC,LaBella2021SupervisedMP,Safdarian2021CoalitionalGT,Hadjidimitriou2021MathematicalOF}.
Graph theory has been proposed as a methodology to efficiently identify optimal design of electricity networks \cite{Nguyen12,Yamagata2015ProposalFA,Tanjo2016GraphPO}. 

Recent advances in both quantum hardware and algorithm development have made it possible to solve several optimization problems on modern quantum computers, with particular success for problems that can be mapped onto graphs. In general, quantum computing has the potential to solve optimization problems \cite{Hen2016DriverHF,PhysRevApplied.5.034007,Zahedinejad2017CombinatorialOO,Hadfield2019FromTQ}. Combinatorial optimization problems (especially NP-hard problems) are of particular interest, since for many of these problems efficient classical algorithms are unknown. 

One of the quantum computing models which can be used for finding solutions to combinatorial optimization problems is the Adiabatic Quantum Optimization model (AQO). AQO refers to the use of the adiabatic theorem in quantum mechanics \cite{griffiths2018introduction, Messiah2014QuantumM} to adiabatically move towards the ground state of an interaction Hamiltonian onto which an optimization problem is mapped. The architecture which approximately implements the adiabatic quantum computer paradigm is the quantum annealer \cite{McGeoch2020TheoryVP}. Differently from a quantum gate architecture, on which a computation is defined by applying separate gates (operators) to each single qubit, a quantum annealer defines a Hamiltonian operator which acts simultaneously on all qubits in the ground state and adiabatically changes this operator towards a Hamiltonian on which the cost function of a specific optimization problem is mapped \cite{Lucas_NP_Ising}. As a result of the quantum mechanics adiabatic theorem, the theoretical resulting state after the adiabatic evolution will be the ground state of the final Hamiltonian, i.e. the global minimum of the optimization problem mapped onto the final Hamiltonian.

In this paper, we focus on solving the power network partitioning problem on a quantum annealing architecture. 
We choose to follow the graph partitioning problem formulation presented in \cite{Tanjo2016GraphPO}, where a graph partitioning approach is used to calculate the optimal transfer of electricity surplus within a power grid. First, we test our graph partitioning model and implementation with synthetic data sets on quantum and hybrid implementations. Then we apply the hybrid approaches to a large-size problem, optimizing the electricity surplus of the transmission power network of Germany. Thus, we solve the graph partitioning problem for this network with both hybrid and classical approaches. 

This article is structured as follows. In \autoref{sec:qubo-formulation} we outline the mathematical framework at the basis of the optimization problem we aim to solve. Then, in \autoref{sec:graph-partitioning-problem-definition} we define the power network optimization problem, with and without electricity sharing among nodes of the network. In \autoref{sec:hamiltonian} we express the problem in the QUBO formulation, which will be used in the actual implementation and run on the quantum and classical solvers. In \autoref{sec:numerical-results} we present the results of the optimization of the German transmission power network with the different quantum and classical solvers. Finally, in \autoref{sec:conclusions} we summarize our findings and outline possible future developments.

\section{QUBO Formulation and mathematical toolbox}\label{sec:qubo-formulation}
In order to solve combinatorial optimization problems with quantum annealing hardware, one must map the cost function of the optimization problem to a Hamiltonian.
The most common way to do this is to use the Quadratic Unconstrained Binary Optimization (QUBO) formulation \cite{venegas2018cross}.
The QUBO formulation is easily transformed to describe a Hamiltonian and therefore most software suites which drive the quantum annealing hardware can use the QUBO formulation directly as input \cite{venegas2018cross}. 
The QUBO formulation has the following form:
\begin{equation}
    \min_{\bm{x}\in\{0,1\}^n} \bm{x}^\T Q \bm{x},
\end{equation}
where Q is a real upper triangular matrix and $\bm{x}$ and its transpose $\bm{x}^\T$ are binary vectors.

Most problems are not naturally expressed in the QUBO formulation, which requires the recasting of the problem.
Many problems have already been reformulated into the QUBO form \cite{lucas2014ising}.
For those problems that are not yet in such a list, techniques have been distilled to rewrite combinatorial optimization problems into the QUBO form \cite{glover2022quantum}.
The specific techniques used in this paper are: (1) including equality constraints, (2) including inequality constraints and (3) pairwise degree reduction.

\subsection{Including Equality Constraints} \label{sec:eq_penalty}
The QUBO formulation does not allow for any constraints. Nonetheless, many real world problems contain (in)equality constraints.
These constraints are included into the objective function by means of a penalty function \cite{glover2022quantum}.
A penalty function for a constraint should be zero when the constraint is met and positive otherwise.

For the sake of the explanation, let us construct such a penalty function for the linear equality constraint shown in \autoref{eq:lin_eq}
\begin{equation}
\label{eq:lin_eq}
 \bm{a}^\T\bm{x} = b,
\end{equation}
where $\bm{a}$ is a real constant vector and $b$ is a real constant.
If we take the left-hand side, subtract $b$ from it and square it we get
\begin{equation}
P(\bm{x}) = \left(\bm{a}^\T\bm{x} - b\right)^2.
\end{equation}
Note that $P(\bm{x})$ is non-negative and that it is zero if and only if the constraint from \autoref{eq:lin_eq} is met.
Therefore, the addition of this function to the objective function will favour solutions that comply with the constraints in \autoref{eq:lin_eq} (when assuming minimization).

\subsection{Including Inequality Constraints}\label{sec:ineq_penalty}
Including inequality constraints into the objective function requires a bit more work compared to equality constraints.
A general linear inequality constraint has the following form:
\begin{equation}
\label{eq:lin_ineq}
 \bm{a}^\T\bm{x} \le b.
\end{equation}
If $\bm{a}$ and $b$ are integer valued, then this equation is equivalent to an equality constraint with the inclusion of slack variables \cite{glover2022quantum}.
Subsequently, the latter equality constraint can be included into the objective function with the penalty function shown in the previous section.

However, in many real life problems (including the one proposed in this paper), $\bm{b}$ and $a$ in \autoref{eq:lin_ineq} are real valued.
In this case we cannot transform the inequality constraint to an equality constraint.
Instead we directly construct an (approximate) penalty function from the inequality constraint.
We start by finding a (possible negative) lower bound of the left hand side $c\in \R$: 
\begin{equation}
\label{eq:lin_ineq2}
 c\le\bm{a}^\T\bm{x} \le b.
\end{equation}
In general, $c$ can be easily computed. If $\bm{a}$ contains at least one negative element, then $c$ is the sum of all negative elements of $\bm{a}$, otherwise it is $0$.
Next, we want to approximate \autoref{eq:lin_ineq2} using $K$ auxiliary binary variables $z_i$.
The \emph{approximate} penalty function of \autoref{eq:lin_ineq2} is given by
\begin{equation}
\label{eq:lin_ineq_penalty}
P(\bm{x},\bm{z}) = 
\left( \frac{2^K-\frac{1}{2}}{b-c}\left(\bm{a}^\T\bm{x} - c\right) - \sum_{i=0}^{K-1}2^i z_i\right)^2.
\end{equation}
For any $\bm{x}$ that complies with the constraint in \autoref{eq:lin_ineq2} there exists a $\bm{z}$ such that $P(\bm{x},\bm{z})\le \frac{1}{4}$.
For all other $\bm{x}$ we have $P(\bm{x},\bm{z})>\frac{1}{4}$ for all $\bm{z}$.
Hence, when we add this penalty function to the objective function, solutions that meet the constraint of \autoref{eq:lin_ineq} are favoured over solutions that violate the constraint.
A more rigorous in-depth analysis and the derivation of \autoref{eq:lin_ineq_penalty} can be found in \autoref{sec:proofs}.

\subsection{Degree Reduction}\label{sec:degree_reduction}
It is not uncommon that problem formulations are not quadratic, but contain higher order terms.
In such cases, the higher order terms can be reduced to quadratic terms by adding auxiliary binary variables.
For pairwise degree reduction, we substitute a quadratic $x_ix_j$ term by a new variable $z$ and we add the following penalty function to the objective
\begin{equation}
M(x_i,x_j,z) = x_ix_j -2z(x_i+x_j) + 3z.
\end{equation}
Note that $M(x_i,x_j,z)$ is non-negative and $M(x_i,x_j,z)=0$ if and only if $x_ix_j=z$.
So for example $x_1x_2x_3$ would become $x_1z + \lambda M(x_2,x_3, z)$, where $\lambda > 0$ is a Lagrange multiplier.

\section{Graph partitioning problem for power grids}\label{sec:graph-partitioning-problem-definition}

Graphs models are extensively used across a broad span of disciplines to model different applications. 

Specifically, Graph Partitioning (GP), i.e., the determination of communities, or clusters, within the nodes of a graph has proven to be a useful method in network analysis \cite{Bichot2013GraphPB}. In fact, GP approaches emerged to reduce the complexity of applications generally involving graph sizes arbitrarily large compared to the computational resources at hand, by dividing (\textit{partitioning}) the graph into smaller sub-graphs or sub-problems, and thus increase the computational performance \cite{Kobayashi2011}. 

Applications of GP include physical network design \cite{Rosato2021HeuristicGP}, VLSI design \cite{Gottschalk2016VlsiPD}, telecommunication network design \cite{telecommunicationGraphPartitioning}, load balancing of high performance computing (HPC) codes \cite{10.1007/978-3-319-27308-2_33}, distributed sparse matrix-vector multiplication \cite{SparseMatrixMultiplicationGraphPartitioning}, biological \cite{Navlakha2010ExploringBN} and social networks \cite{Tsourakakis2014FENNELSG,Lopes2020APG}.

GP is an NP-hard problem (as a decision problem) \cite{HyafilRivest1973,GAREY1976237}. Algorithms to solve GP problems include exact algorithms, spectral partitioning, geometric partitioning, flow computations, etc. \cite{Bulu2016RecentAI}. Exact algorithms are often used to solve small-size problems, whereas heuristic algorithms are needed for larger problems.

Recently, quantum computers have been used to model and solve the graph partitioning problem, see e.g., \cite{UshijimaMwesigwa2017GraphPU,Pramanik2020QuantumAssistedGC,Negre2020DetectingMC,Chukwu2020ConstrainedoptimizationAD}. 
In this paper, we shall focus on solving the graph partitioning of a power network with quantum annealing, which uses quantum physics to find low-energy states of a problem and which can be mapped to the optimal or near-optimal solution of the optimization problem.

In the following we outline the mathematical formulation of the graph partitioning problem for a power grid, in which the performance metric is given by the \textit{power grid cost}. We follow the formulation presented in \cite{Tanjo2016GraphPO}. 
We first consider the case of no electricity sharing within partitions and then extend this formulation to allow for electricity sharing.
Finally, we shall express the graph partitioning problem in the Hamiltonian/QUBO framework needed to implement and solve the problem on a quantum annealer.

\subsection{Mathematical formulation without electricity sharing}\label{subsec:graph-partitioning-math-formulation}
In \cite{Tanjo2016GraphPO}, the graph partition of power grids is modeled in the following manner. 
Suppose we have a graph $G=(V,E)$, where the vertices $V$ represent geographical areas and the the edges $E$ represent transmission lines between two areas.
The total electricity surplus in an area is then encoded in the weight of the vertices associated to the given area. 

Let $P$ be the number of partitions, or clusters, and $N=|V|$ be the number of vertices in the graph.
Then the binary variable $v_{np}\in \B$ is $1$ when vertex $n$ is in partition $p$ and $0$ otherwise.
Let $k$ denote the threshold of self-sufficiency of the partition, as defined in \cite{Tanjo2016GraphPO}.
Suppose $\alpha,\beta \in \R^+$ are coefficients of the cost for transmission lines within a single cluster and between clusters respectively. Then the CQM becomes:
\begin{eqnarray}
\min_{\bm{v}\in\B^{PN}}  \sum_{p=1}^P\alpha \left(\sum_{n=1}^Nv_{np}\right)^2 + \beta\left(|E|-\sum_{\{n,m\}\in E}v_{np}v_{mp}\right) \label{eq:obj_original}\\
\text{subject to:}\qquad\qquad           \sum_{p=1}^P v_{np}  =1, \quad   n=1,\ldots,N \label{eq:onehot_constraint}\\
\sum_{n=1}^Nv_{np}w_n - k\sum_{n=1}^Nv_{np}\le 0, \quad p=1,\ldots,P.\label{eq:balancing_constraint}
\end{eqnarray}
\autoref{eq:obj_original} is called the objective function of the optimisation problem, which represent the cost of constructing a particular partition.
The constraints described in \autoref{eq:onehot_constraint} is called the one-hot constraint.
This constraint enforces each vertex to be in one and one only cluster.
Lastly, the constraint in \autoref{eq:balancing_constraint} is called the balancing constraint.
The latter ensures that the average surplus of a cluster is below the threshold $k$, so that the partition fulfills the \emph{self-sufficiency} property.

\subsection{Mathematical formulation with electricity sharing}\label{sec:electricity-sharing}
The mathematical formulation of the problem with electricity sharing is very similar to the problem outlined in the previous section.
Both the objective function and one-hot constraint remain the same.
The key difference is in the balancing constraint, which will allow for the sharing.
This is done by adding a flow function between each partition $F(\bm{v};p,q)$, such that $F(\bm{v};p,q)=-F(\bm{v};q,p)$.
The resulting mathematical formulation is then given by
\begin{eqnarray}
\min_{\bm{v}\in\B^{PN}}  \sum_{p=1}^P\alpha \left(\sum_{n=1}^Nv_{np}\right)^2 + \beta\left(|E|-\sum_{\{n,m\}\in E}v_{np}v_{mp}\right) \label{eq:obj_original2}\\
\text{subject to:}\qquad\qquad           \sum_{p=1}^P v_{np}  =1, \quad   n=1,\ldots,N \label{eq:onehot_constraint2}\\
\sum_{n=1}^Nv_{np}w_n + \sum_{q=1}^PF(\bm{v};p,q)- k\sum_{n=1}^Nv_{np}\le 0, \quad  p=1,\ldots,P.\label{eq:balancing_constraint2}
\end{eqnarray}

\section{QUBO formulation of graph partitioning}\label{sec:hamiltonian}
In this paper we solve the GP problem for power grids using the D-Wave CQM solver, D-Wave BQM solver, quantum annealing and Microsoft Azure solvers.
The CQM solver can take the problem as given in the previous sections.
However, to use the other solvers, the problem must be rewritten to the QUBO formulation.
The optimisation problem described in \autoref{sec:graph-partitioning-problem-definition} contains constraints, which are not allowed in the QUBO formulation.

To overcome this, the constraints will be included into the objective function with the use of penalty terms to produce the following form:
\begin{equation}
\min H(\bm{x}) + \sum_{i=1}^{\#\text{constraints}}\lambda_iP_i(\bm{x}).
\end{equation}
We will first show the QUBO of the GP problem \emph{without} electricity sharing, followed by the QUBO formulation of the GP problem \emph{with} electricity sharing.

\subsection{QUBO formulation of graph partitioning without electricity sharing}\label{sec:qubo-hamiltonian}
First, we derive the Hamiltonian for the objective function.
The objective function shown in \autoref{eq:obj_original} consists of binary variables and is quadratic in nature.
Hence, this objective function already complies with QUBO formalism and its  Hamiltonian is given by
\begin{equation}
    \label{eq:obj_ham}
        H(\bm{v}) = \sum_{p=1}^P\alpha \left(\sum_{n=1}^Nv_{np}\right)^2 + \beta\left(|E|-\sum_{(n,m)\in E}v_{np}v_{mp}\right).
\end{equation}

Secondly, we will derive a penalty function for the one-hot constraint shown in \autoref{eq:onehot_constraint}.
Since it is a linear equality constraint, we can use the technique described in \autoref{sec:eq_penalty} to produce
\begin{equation}
\label{eq:onehot_penalty}
P_{\text{oh}}(\bm{v}) = \sum_{n=1}^N\left(\sum_{p=1}^P v_{np}  - 1\right)^2.
\end{equation}

Lastly, we will derive a penalty function for the balancing constraint shown in \autoref{eq:balancing_constraint}.
To do so, we will use the technique described in \autoref{sec:ineq_penalty} for which we will need a lower bound on the constraint.
A lower bound can be found by taking the summation over all negative $w_n-k$ values, i.e.,
\begin{equation}
    \label{eq:nonnegative_constant}
    c = \frac{1}{2}
    \sum_{n=1}^N\left(w_n-k-|w_n-k|\right).
\end{equation}
This lower bound $c$ and the upper bound of $0$ produces the following penalty function:
\begin{equation}
\label{eq:balance_penalty1}
P_\text{bc}(\bm{v},\bm{x}; K) =  \sum_{p=1}^P\left(\frac{2^K-\frac{1}{2}}{-c} \left(-c + \sum_{n=1}^Nv_{np}(w_n - k)\right) - \sum_{a=0}^{K-1}2^ax_{ap}\right)^2,
\end{equation}
where $x_{ap}$ are the auxiliary (slack) variables.

The final QUBO Hamiltonian representation of the whole problem is then
\begin{equation}
\label{eq:hamiltonian}
    \min_{(\bm{v},\bm{x})\in \B^{PN}\times\B^{PK}}
    H(\bm{v})
    +\lambda_\text{oh}P_\text{oh}(\bm{v})
    +\lambda_\text{bc}P_\text{bc}(\bm{v},\bm{x}; K),
\end{equation}
where $\lambda_\text{oh}$ and $\lambda_\text{bc}$ are non-negative constants (Lagrange multipliers) and $K$ is a positive integer.

From \autoref{eq:hamiltonian}, it is evident that the number of variables for the problem with $N$ nodes, $P$ partitions and  the hyperparameter $K$ is $P(N + K)$.

\subsection{QUBO formulation of graph partitioning with electricity sharing}
The mathematical formulation of the GP problem with sharing has the same objective function and one-hot constraint. Therefore, we will use the same objective Hamiltonian and penalty function as in the previous section.
What remains is constructing a penalty for the new balancing constraint.
However, as will become clear, this altered balancing Hamiltonian will need new auxiliary variables for which the behaviour will be enforced by an additional constraint.

Let us start by defining a flow from partition $p$ to partition $q$.
We will construct a flow which can share everything from $p$ to $q$, everything from $q$ to $p$ or nothing at all.
This flow function has the following form:
\begin{equation}
\label{eq:flow}
   F(\bm{v}, \bm{f}; p, q) = \sum_{\{n,m\}\in E}w_{nm}(f_{nmp} - f_{nmq})(v_{np}v_{mq} + v_{mp}v_{nq}),
\end{equation}
where $f_{nmp}$ are new binary variables determining the flow between partition $p$ and $q$.
Note that $v_{np}v_{mq} + v_{mp}v_{nq}$ in \autoref{eq:flow} is 1 if $\{n,m\}$ is an edge between partition $p$ and $q$ and is zero otherwise.

Note that $F$ is indeed a flow, since
\begin{align}
    F(\bm{v}, \bm{f}; q, p) &=\sum_{\{n,m\}\in E}w_{nm}(f_{nmq} - f_{nmp})(v_{nq}v_{mp} + v_{mq}v_{np})\\
    &=-\sum_{\{n,m\}\in E}w_{nm}(f_{nmp} - f_{nmq})(v_{np}v_{mq} + v_{mp}v_{nq})\\
    &=-F(\bm{v}, \bm{f}; p, q).
\end{align}
Therefore, we also have the property $F(\bm{v}, \bm{f}; p, p)=0$.

If we include $F$ from \autoref{eq:flow} into the balancing constraint, the constraint becomes cubic in nature.
Hence, we will have to reduce the degree of this constraint by using the technique described in \autoref{sec:degree_reduction}.
Let $\mathcal{P} =\{(p,q)\in \{1,\ldots,P\}^2| p\not=q\}$, then we reduce the degree by adding auxiliary binary variables with the following properties
\begin{align}
a_{nmpq} = v_{np}v_{mq} &, \forall \{n,m\}\in E, \forall (p,q) \in \mathcal{P},\\
y_{nmpq} = f_{nmp}a_{nmpq}  &, \forall \{n,m\}\in E, \forall (p,q) \in \mathcal{P},\\
z_{nmpq} = f_{nmp}a_{nmqp}  &, \forall \{n,m\}\in E, \forall (p,q) \in \mathcal{P}.
\end{align}
Hence, $F(\bm{v}, \bm{f}; q, p)$ becomes
\begin{align}
 F(\bm{v}, \bm{f}; p, q) &=  \sum_{\{n,m\}\in E}w_{nm}(f_{nmp} - f_{nmq})(v_{np}v_{mq} + v_{mp}v_{nq})\\
 &=\sum_{\{n,m\}\in E}w_{nm}(f_{nmp} - f_{nmq})(a_{nmpq} + a_{nmqp})\\
 &=\sum_{\{n,m\}\in E}w_{nm}(y_{nmpq}+z_{nmpq} - y_{nmqp}-z_{nmqp})\\
 &= G(\bm{y},\bm{z};p,q)
\end{align}

The penalty function of the balancing constraint with electricity sharing is then given by
\begin{equation}
P_\text{bc}(\bm{v}, \bm{x}, \bm{y}, \bm{z}) =  \sum_{p=1}^P\Bigg(\frac{2^K-\frac{1}{2}}{c} \Big(c + \sum_{n=1}^Nv_{np}(w_n - k)\Big)
+\sum_{\substack{q=1 \\ q\not=p}}^PG(\bm{y},\bm{z}; q,p)
- \sum_{a=0}^{K-1}2^ax_{ap}\Bigg)^2~,
\end{equation}
where $c = \frac{1}{2}\sum_{n=1}^N\left(|w_n-k|- w_n+k\right) + \sum_{\{n,m\}\in E}w_{nm}.$

Lastly, we apply the penalty function for degree reduction as described in \autoref{sec:degree_reduction} to imply the required behaviour of the auxiliary variables:
\begin{align}
P_{aux}(\bm{v},\bm{f},\bm{y},\bm{z},\bm{a}) = \sum_{\{n,m\}\in E}\sum_{p,q\in\mathcal{P}}M(v_{np},v_{mq}, a_{nmpq}) + M(f_{nmp},a_{nmpq}, y_{nmpq}) + M(f_{nmp},a_{nmqp}, z_{nmpq}).
\end{align}

With the new penalty terms, the QUBO optimization problem for graph partitioning with electricity sharing is given by:
\begin{equation}
\min_{\bm{v},\bm{x},\bm{f}, \bm{y},\bm{z}, \bm{a}}
H(\bm{v})
+\lambda_\text{oh}P_\text{oh}(\bm{v})
+\lambda_\text{bc}P_\text{bc}(\bm{v}, \bm{x}, \bm{y}, \bm{z}; K)
+\lambda_\text{aux}P_{aux}(\bm{v},\bm{f},\bm{y},\bm{z},\bm{a}).
\end{equation}
The total number of variables of the QUBO formulation is less straightforward compared to the model without electricity sharing.
We shall analyse this by counting the size of all the binary vectors $\bm{v},\bm{x},\bm{f}, \bm{y},\bm{z}$ and $\bm{a}$.
The size of $\bm{v}$ and $\bm{x}$ is the same as for the problem without sharing, $PN$ and $PK$ respectively.
For each edge there is a $f$ variable for each partition, hence the size of $\bm{f}$ is $P|E|$.
Lastly, the size of $\bm{a}$, $\bm{y}$ and $\bm{z}$ is given by $\mathcal{P}$, which is $|E|(P^2-1)$.
Therefore, the total number of variables for the QUBO with electricity sharing is $P(N+K+|E|) + 3|E|(P^2-1)$.
Hence, when the network graph is not sparse, the number of variables is much larger for the model with sharing.

\subsection{QUBO Hyperparameters}\label{sec:qubo-hyperparameters}
In the construction of the QUBO in the previous sections, the hyper-parameters $K$ and $\lambda_i$ were introduced.
These hyper-parameters should be determined before solving the problem and comparing results.
The $K$ parameter represent the level of precision we desire in the approximation of the balancing constraint, where higher $K$ corresponds to a higher precision.
A more in-depth discussion on the inequality constraints and the effects on $K$ can be found in \autoref{sec:proofs}.

The other hyper-parameters, $\lambda_i$, were determined by a grid-search.
In this grid-search each grid point was evaluated using short runs of simulated annealing.
For each grid point, two measures of quality were calculated: (i) the original objective value of \autoref{eq:obj_original} and (ii) the number of constraints violated.
The best grid point on the grid is then the point were this objective is the lowest, whilst no constraints were violated.
To minimize the computation time, we performed the grid-search in two stages: (i) a logarithmic stage and (ii) a linear stage.
In the logarithmic stage a logarithmic grid was searched quickly to determine to order of magnitude of each $\lambda_i$.
Next, a linear grid was searched to more precisely determine the best value of $\lambda_i$.
Using this method we could quickly find suitable hyper-parameters $\lambda_i$.

\subsection{Implementation on D-Wave architecture}\label{sec:implementation}
The D-Wave system \footnote{\url{https://www.dwavesys.com/}} is a hardware heuristic that minimizes Ising objective functions using a physically realized version of quantum annealing. Due to the mapping between QUBO and Ising variables \cite{Lucas_NP_Ising}, every QUBO problem can be translated to an Ising model which can then be embedded onto a D-Wave system.


The implementation of the Hamiltonian in \autoref{eq:hamiltonian} has been written in Python and run on both the D-Wave Advantage\textsuperscript{TM} QPU and D-Wave’s Leap hybrid solver service, which uses the D-Wave Advantage\textsuperscript{TM} system as a back-end.
The hybrid solver contains a portfolio of heuristic solvers that leverage quantum and classical methods to solve problems much larger than can fit on Advantage quantum systems.
To date, the service includes three solvers: (i) Leap Hybrid Solver, to solve binary quadratic models (BQMs); (ii) Leap Hybrid discrete quadratic model (DQM) solver, for problems on categorical variables; and (iii) Leap Constrained Quadratic model (CQM) solver, which extend the previous two and allows for expressing constraints arithmetically.

For the numerical results we focus on the case without electricity sharing.

\section{Numerical results}\label{sec:numerical-results}
The goal of the paper is to determine the performance of quantum annealing assisted solvers (BQM and CQM) in solving the graph partitioning problem for energy grids in relation to other approaches. In particular, for this comparison we consider the solvers available in the Microsoft Quantum-inspired optimization (QIO) provider \footnote{https://learn.microsoft.com/en-us/azure/quantum/provider-microsoft-qio} present on the Azure cloud. 

\subsection{Testing CQM and BQM implementations}\label{sec:testing-results}
To gain confidence in both our implementation and derivation of the CQM and BQM models, they were first tested on smaller problems.
For this purpose, we designed two test.

Firstly, the implementation was tested.
This was done by constructing a square graph with known weights.
Because of the size of the problem, it was possible to write both the CQM and BQM models out by hand.
These were then compared to the CQM and BQM models of our implementation and were found to be equal.
Hence, we gained confidence in the implementation of our model.

Secondly, the equality of the optimal points of the CQM and BQM in practice was tested.
In theory, the optimal point of the CQM and BQM should give the same value when plugged into the original objective function \autoref{eq:obj_original}.
We tested this by means of brute force algorithms.
We tested equality of the original objective for 100 randomly generated graphs with 2 to 8 nodes.
All 700 tested problems gave a solution with equal objective value for both BQM and CQM.
Hence, we gained confidence on the equivalence of the optimal points for the BQM and CQM.

\subsection{Synthetic data set}

With the confidence in the implementation of BQM and CQM, the performance of the BQM using quantum annealing on a real QPU was tested.
To investigate the performance, a data set was created by connecting two cliques with a single edge (see \autoref{fig:clique_graphs}).
In this way, a graph with two clearly defined partitions is created (the two cliques).
Next, we added a surplus to nodes of graph.
The surplus we added has two key properties:
\begin{enumerate}
    \item The average surplus in each clique is $0.45$.
    \item More then 1/3 of all possible partitions violate the balancing constraint with a threshold $k=0.5$.
\end{enumerate}
Properties (1) ensures that the balancing constraint is not violated when the two cliques are selected as partitions, while property (2) ensures that there exists combinations of variables were the balancing constraint can be violated.

The chosen data set has two key advantages; it has a variable size and the optimum is known beforehand.
Hence, we can determine if we find the global optimum for different problem sizes.
In \cite{Tanjo2016GraphPO} two different sets of cost parameters are considered, $(\alpha, \beta)=(1,1)$ and $(\alpha, \beta)=(1,10)$.
Since the focus of the present analysis lies in a proof of concept of the model on the QPU, the specific value of these parameters is irrelevant.
Hence, we choose $\alpha$ and $\beta$ to be respectively 1 and 10.
The threshold $k$ was set to $0.5$.

The problem was tested directly on the D-Wave Advantage system.
The annealing time was set to \unit[1000]{$\mu$s} and the number of reads was set to 500.
All other settings were left in the their defaults.
The problem was tested for clique sizes 3 to 53, after which no minor-embedding could be found.
The quantum annealer found the optimal solution in \emph{all} instances at least once.
\autoref{fig:qpu_results} shows on the left axis the end to end time of solving the problem.
Since the actual time spend on the QPU was the same for all these problems, almost all this time is spend in finding an appropriate minor-embedding.
On the right axis we see the embedding size as a function of clique size.
Interestingly enough, no minor-embedding could be found for a clique size larger then 53.
This is unexpected as the  minor-embedding for a clique size of 3974, while the D-Wave Advantage system has 5640 qubits.

\begin{figure}
\centering
\begin{subfigure}{.4\textwidth}
  \centering
  \include{images/clique_graphs}
  \caption{Shape of the graphs used for the QPU results.}
  \label{fig:clique_graphs}
\end{subfigure}%
\begin{subfigure}{.6\textwidth}
  \centering
  \begin{tikzpicture}
\begin{axis}[
  axis y line*=left,
  xmin=3, xmax=53,
  ymin=0, ymax=800,
  xlabel={Clique Size [\#Nodes]},
  ylabel={E2E Time [s]},
]
\addplot[smooth,mark=x,red]
  coordinates{
    (3, 0.19649)
    (4, 0.37886)
    (5, 0.53844)
    (6, 0.59578)
    (7, 0.92738)
    (8, 1.33093)
    (9, 2.93722)
    (10,2.35356)
    (11,3.69271)
    (12,4.11662)
    (13,5.37846)
    (14,5.53185)
    (15,9.43014)
    (16,14.3543)
    (17,23.7875)
    (18,19.5713)
    (19,33.9270)
    (20,14.4864)
    (21,12.5786)
    (22,17.5623)
    (23,21.9316)
    (24,20.0077)
    (25,21.9526)
    (26,22.6800)
    (27,49.8193)
    (28,18.2332)
    (29,70.3598)
    (30,70.5463)
    (31,34.3013)
    (32,60.5408)
    (33,115.227)
    (34,88.4286)
    (35,123.205)
    (36,79.8793)
    (37,171.372)
    (38,103.147)
    (39,99.6073)
    (40,106.470)
    (41,140.297)
    (42,262.107)
    (43,324.605)
    (44,155.159)
    (45,229.293)
    (46,172.335)
    (47,211.820)
    (48,345.804)
    (49,727.947)
    (50,150.900)
    (51,215.793)
    (52,318.288)
    (53,320.995)
}; \label{plot_one}

\end{axis}

\begin{axis}[
  axis y line*=right,
  axis x line=none,
  xmin=3, xmax=53,
  ymin=0, ymax=5000,
  ylabel={Embedding Size [\#Qubits]},
  legend pos = north west
]
\addlegendimage{/pgfplots/refstyle=plot_one}\addlegendentry{E2E Time}
\addplot[smooth,mark=*,blue]
  coordinates{
    (3,48)
    (4,68)
    (5,82)
    (6,105)
    (7,123)
    (8,152)
    (9,173)
    (10,226)
    (11,263)
    (12,272)
    (13,319)
    (14,367)
    (15,376)
    (16,418)
    (17,461)
    (18,555)
    (19,544)
    (20,681)
    (21,759)
    (22,758)
    (23,856)
    (24,998)
    (25,998)
    (26,1117)
    (27,1204)
    (28,1336)
    (29,1265)
    (30,1346)
    (31,1486)
    (32,1581)
    (33,1548)
    (34,1628)
    (35,1695)
    (36,1900)
    (37,1975)
    (38,2035)
    (39,2370)
    (40,2294)
    (41,2370)
    (42,2478)
    (43,2634)
    (44,2762)
    (45,2964)
    (46,3085)
    (47,3329)
    (48,3401)
    (49,3470)
    (50,3559)
    (51,3841)
    (52,3827)
    (53,3974)
}; \addlegendentry{Embedding Size}
\end{axis}

\end{tikzpicture}

  \caption{On the left axis the E2E time versus clique size. The right axis shows the embedding size (in qubits) versus clique size.}
  \label{fig:qpu_results}
\end{subfigure}
\caption{A figure with two subfigures}
\label{fig:test}
\end{figure}
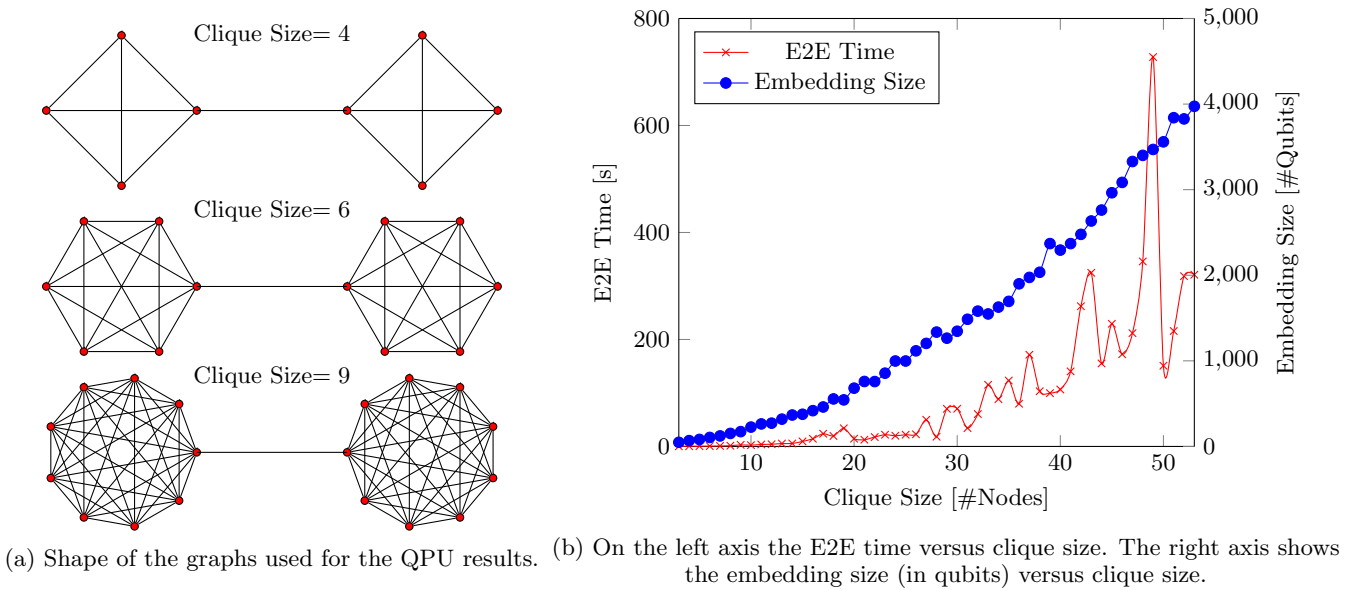

\subsection{German transmission power network data set}\label{sec:germany-results}

Next we performed a benchmarking analysis on a large size energy network. For this analysis, we considered the transmission power data obtained with the open-source reference model of European transmission networks, SciGRID \cite{SciGRIDv0.2}, which is based on the raw transmission
data available in openstreetmap.org \footnote{OpenStreetMap: https://www.openstreetmap.org/}. 

In particular, we performed the analysis on the SciGRID model output tables with vertices and links obtained from running the model on Germany raw data. The SciGRID output tables are: the table \textit{vertices}, which contains the geographical center positions of German electrical substations, together with information on voltage level, frequency, name and operator; the table \textit{links}, which contains the connections between two substations, with information on properties of the transmission line, e.g., voltage, number of cables in the circuit, resistance, and maximum current. 

Given their structure, the SciGRID output tables can be used to easily build the graph of the German transmission power network (\autoref{fig:german-network}) and define the mapping to the problem in Eqs.~(\ref{eq:obj_original}-\ref{eq:balancing_constraint}) and Eqs.~(\ref{eq:obj_original2}-\ref{eq:balancing_constraint2}), where the table \textit{vertices} is used to define the nodes and the table \textit{links} defines the edges of the transmission network. 

Additionally, for the complete definition of the optimization problem, we need to choose the values of the cost parameters $\alpha$ and $\beta$, and the electricity surplus at each node. These parameters are not defined or available in the SciGRID data set. 

In \cite{Tanjo2016GraphPO} two different sets of cost parameters are considered, $(\alpha, \beta)=(1,1)$ and $(\alpha, \beta)=(1,10)$. Since the focus of the present analysis is a benchmark of the quantum optimization approaches, the specific value of these parameters is not relevant. Thus, for the sake of the analysis, we choose $\alpha$ and $\beta$ to be respectively 1 and 10. 

Finally, the electricity surplus at each node (i.e., the weight of the node) is drawn from a uniform $[0,1)$ distribution. In fact, an estimation of a realistic electricity surplus is out of scope for this paper and does not add value to the benchmark, as it would only influence the choice of the hyper-parameters and the value of the objective function.
In our test data, the average surplus across all nodes was approximately $0.49$.
In order to make the problem sufficiently difficult, the threshold was set to $k=0.5$.

\begin{figure}
\centering
\includegraphics[width=0.5\textwidth]{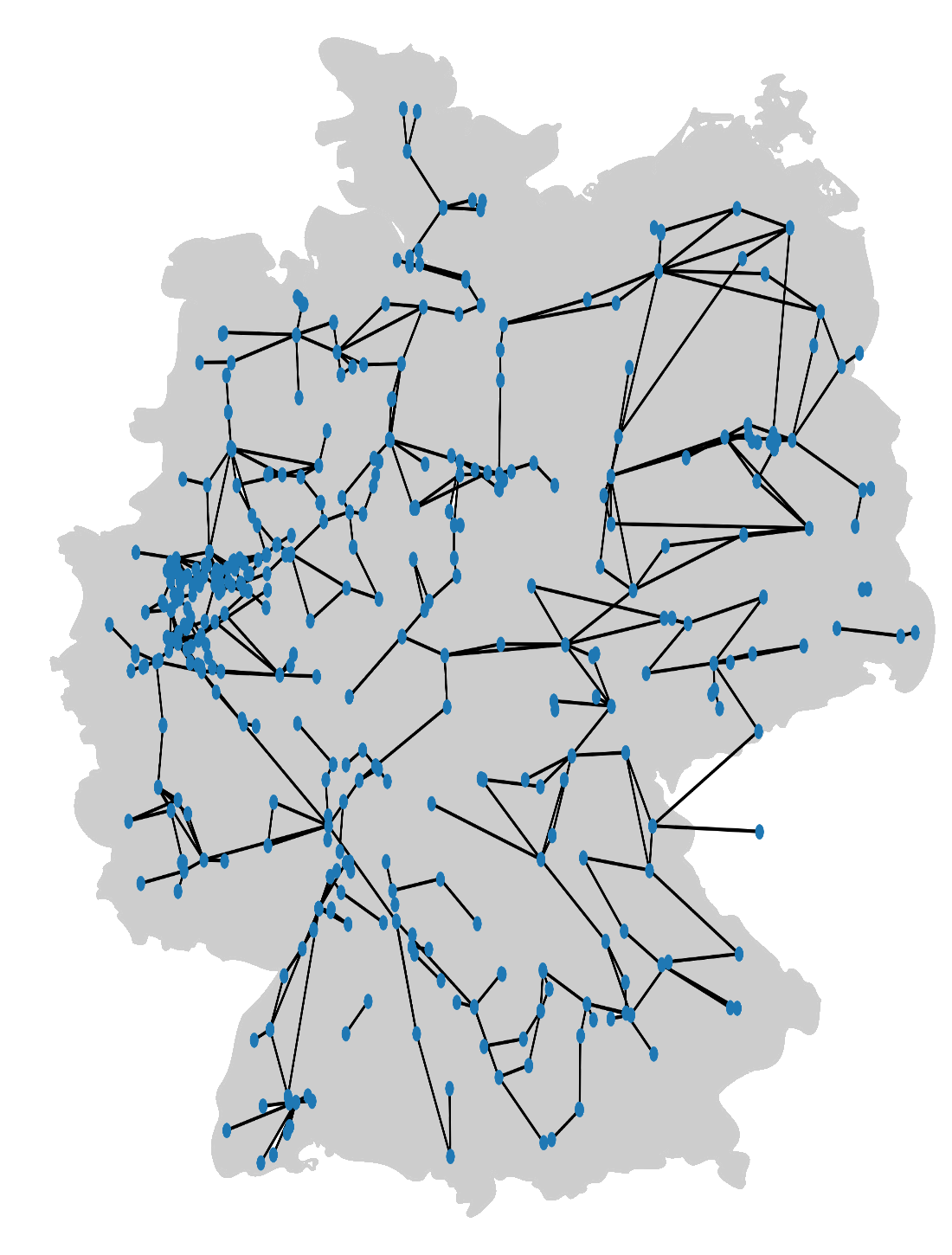}
\caption{Graph derived from the SciGRID German transmission power network.}\label{fig:german-network}
\end{figure}

\autoref{tab:results-hyperparameters} shows the hyper-parameters obtained with the grid-search algorithm outlined in \autoref{sec:qubo-hyperparameters} for the case of the German transmission power network. $K=10$ was chosen, because it was found to be sufficiently large to detect small errors in the balancing constraint, whilst being small enough to be safe in terms of overflow values.
We defined eight different graph partitioning problems, corresponding to different number of partitions (2 to 9) and calculated the hyperparameters for each of these problems. Note that the higher the number of partitions, the higher the number of variables involved becomes, which also drives up the computational time.

\begin{table}[htbp]
    \centering
    \begin{tabular}{rrrr}
    \toprule
    $P$ & $\lambda_{oh}$ & $\lambda_{bc}$\\
    \midrule
    2 &  500 & 0.01\\
    3 &  300 & 0.08\\
    4 &  270 & 0.2\\   
    5 &  250 & 0.2\\   
    6 &  240 & 0.4\\   
    7 &  230 & 0.9\\   
    8 &  230 & 0.9\\   
    9 &  290 & 1.0\\
    \bottomrule
    \end{tabular}
    \caption{Hyper parameters for the QUBO of the German network for different partition sizes $P$, where $\alpha=1$, $\beta=10$, $k=0.5$ and $K=10$}
    \label{tab:results-hyperparameters}
\end{table}

\autoref{table:results} shows the comparison of the minimum of the objective function found by different solvers used to solve the graph partitioning problem on the German transmission power network.

For this comparison we first ran the problem on the CQM and BQM hybrid solvers. Then, we implemented the same binary quadratic problem in the Microsoft QIO framework and solved it by means of Parallel Tempering (PT), Simulated Annealing (SA), Substochastic Monte Carlo (MC) and Tabu Search (TS).

Parallel Tempering and Simulated Annealing map the problem onto a thermodynamic system and search for the optimum by exchanging configurations at different temperatures. Substochastic Monte Carlo is a diffusion Monte Carlo algorithm inspired by adiabatic quantum computation. Tabu Search is a metaheuristic optimization approach based on a local search method, which looks at neighboring configurations to move across the solution space.

For the sake of a fair comparison, all approaches, hybrid and quantum-inspired, are run with a time limit of 10 seconds. Table \ref{table:results} shows the objective value of the optimal solution, thus it provides an estimation of the quality of the solution. The column \textit{P} indicates the number of partitions in which the network is clustered. 

Clearly the BQM and CQM implementations overperform the quantum-inspired methods in terms of quality. For all numbers of partitions, the CQM provides the lowest objective value. 

Additionally, all quantum-inspired methods show a deterioration in the quality of the solution for large number of partitions. This is clearly visible in the behavior of the objective function value for different partition sizes shown in \autoref{fig:solution-deterioration}. In fact, the value of the objective function at the minimum for the classical solvers shows a large variability across different partition sizes, compared to the hybrid solvers (CQM and BQM). For PT, SA, MC and TS the objective function explodes for partition size larger than 5. This is probably due to the fact that hybrid solvers can handle a larger number of variables and constraints.

\begin{table}[htbp]
\centering
    \begin{tabular}{rrrrrrr}
    \toprule
    $P$ & CQM & BQM & PT & SA & MC & TS \\
    \midrule
    2 & 101,362 & 102,232 & 103,022 & 102,842 & 103,845 & 103,484\\
    3 &  74,824 &  75,350 &  78,035 &  78,142 &  77,376 &  77,180\\
    4 &  64,272 &  67,196 &  69,082 &  69,135 &  67,999 &  70,247\\
    5 &  60,130 &  63,762 &  \textbf{65,088} &  \textbf{65,036} &  \textbf{66,218} &  \textbf{67,133}\\
    6 &  \textbf{59,190} &  \textbf{63,110} &  80,095 &  78,808 &  83,875 &  72,645\\
    7 &  60,048 &  64,240 & 136,392 & 139,403 & 134,969 &  79,644\\
    8 &  62,094 &  66,386 & 150,925 & 154,054 & 199,356 &  85,948\\
    9 &  64,882 &  69,180 & 174,162 & 178,692 & 325,600 & 115,515\\
    \bottomrule
    \end{tabular}
\caption{Objective function value for the German power network with different solvers. Each solver has a time limit of 10 seconds.}\label{table:results}
\end{table}

\begin{figure}
\centering
\includegraphics[width=0.75\textwidth]{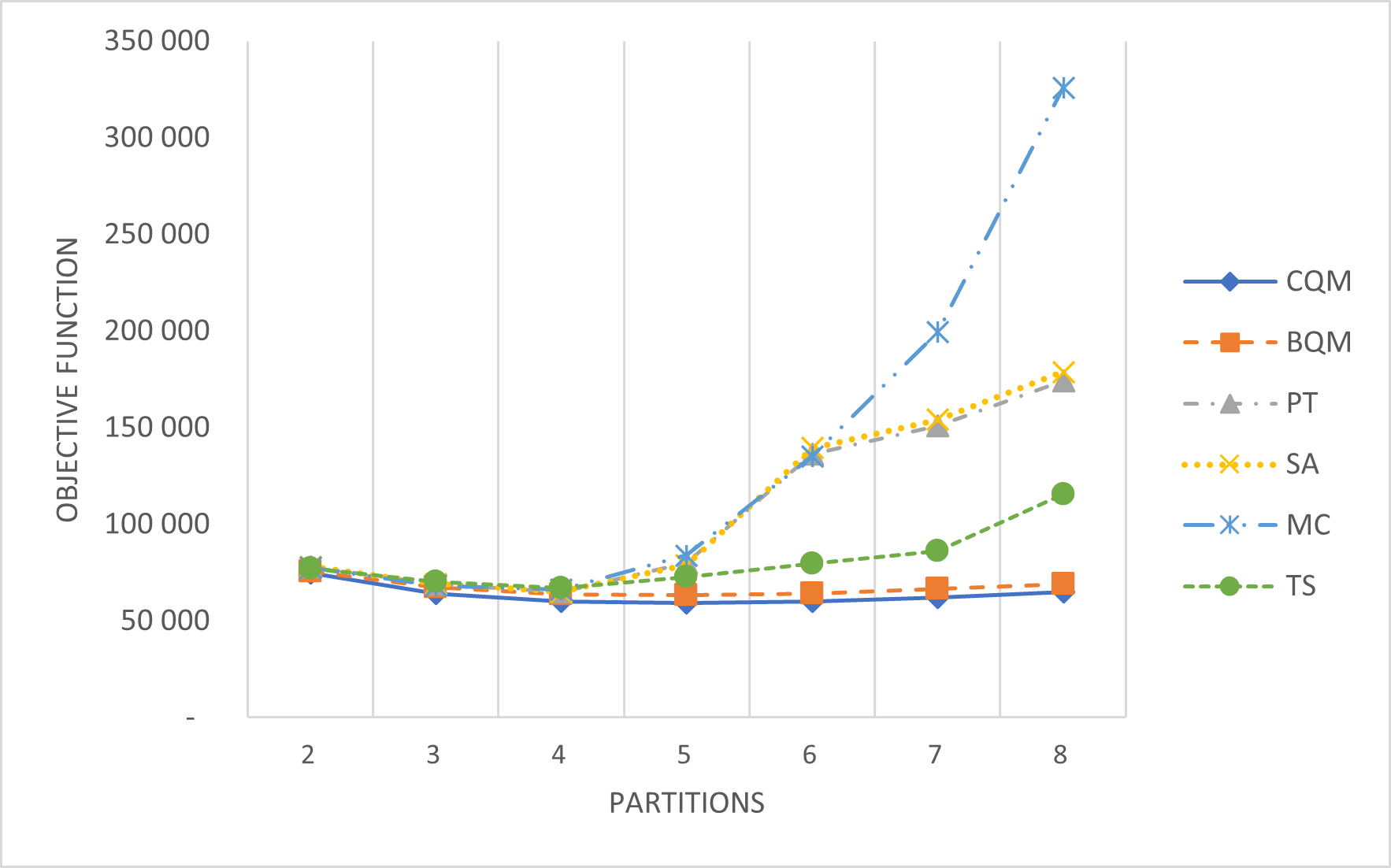}
\caption{Change in objective function value for different solvers.}\label{fig:solution-deterioration}
\end{figure}

\begin{figure}
\centering
\includegraphics[width=0.5\textwidth]{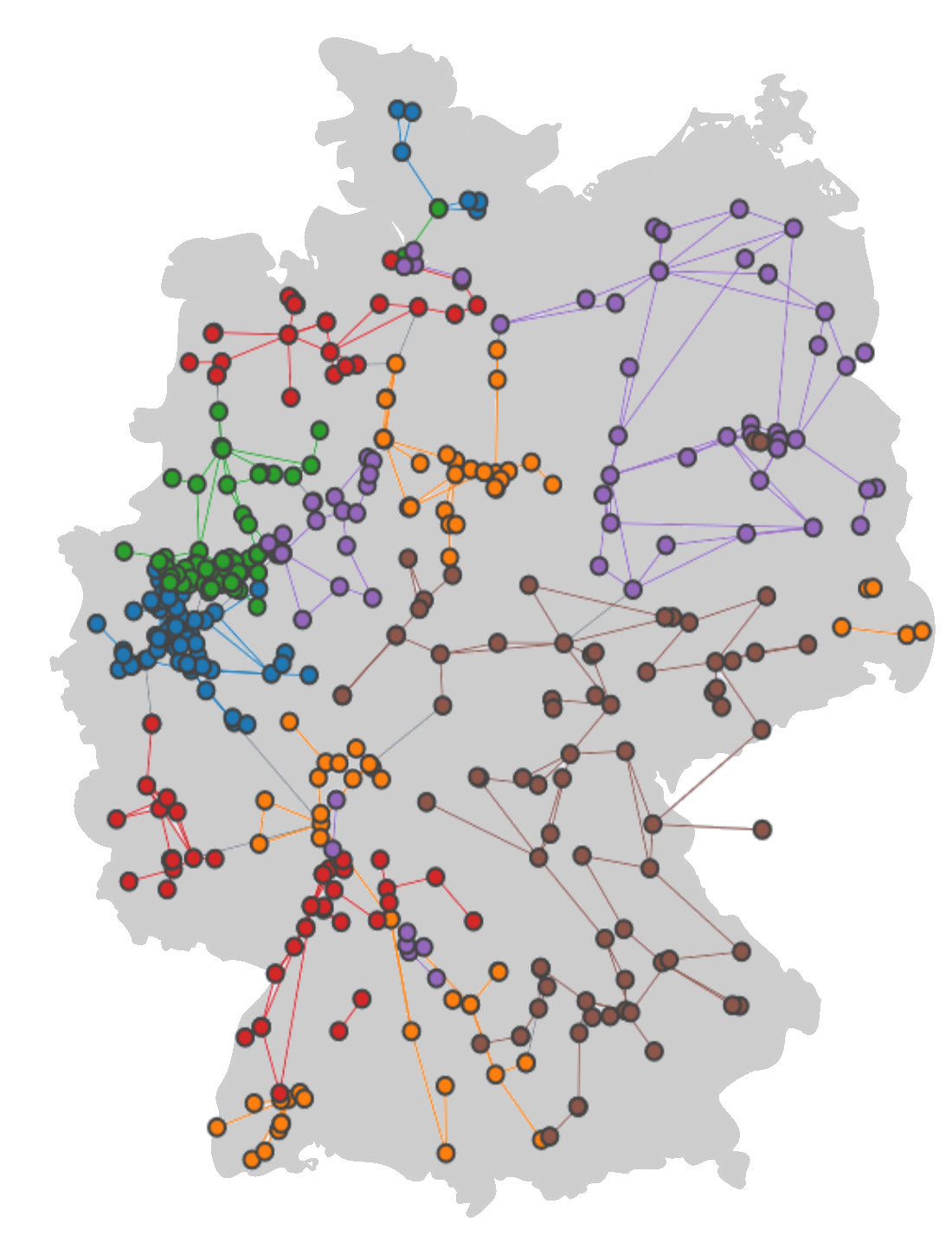}
\caption{Best solution found by the Hybrid CQM Sampler (P=6).}\label{fig:german-network-cqm}
\end{figure}

\autoref{fig:german-network-cqm} shows an example of solution obtained with the Hybrid CQM sampler for $P=6$ partitions. 
Note that some clusters include points that are far from each other. This is the result of the simplification that the edge cost, $\beta$ in \autoref{eq:obj_original}, is a constant. A possible improvement of this assumption would be to define $\beta$ as a function of the geographical distance between nodes.






\section{Conclusions}\label{sec:conclusions}

In this paper we showed how to optimize an energy network with quantum annealing. 

Following \cite{Tanjo2016GraphPO}, we mapped the optimization problem to a graph partitioning problem, in which the optimal set of communities of power substations are determined to share the energy surplus across the network. 

The problem was defined mathematically as a QUBO problem, both with and without electricity sharing across the network. A grid-search approach to determine the optimal hyperparameters of the optimization problem was also introduced. 

The model and the implementation for the case of no-electricity sharing were tested both on a real QPU and on hybrid solvers. The test were carried on small-size random graphs and the results were verified with an exhaustive search algorithm (exact solution).

For the numerical results, we considered a real-size problem, using the German transmission power network data. The optimization problem was implemented for the case of no-electricity sharing and run on the D-Wave hybrid CQM and BQM solvers. 

The results from the quantum solvers were compared to those of classical approaches available on Azure, i.e. Parallel Tempering, Simulated Annealing, Substochastic Monte Carlo and Tabu Search present in the Microsoft QIO framework.

The numerical results showed that, for the same computation time, the quantum approaches (CQM and BQM) outperform the classical algorithms in terms of quality of the solution, as the value of the objective function of the \textit{quantum} solutions is found to be always lower than with the classical approaches across a set of different problem size. 
Further, the quality of the solution of the classical approaches deteriorates with increasing problem complexity (number of variables) much faster than with the quantum approaches. The latter effect is an indication that hybrid-quantum approaches might be a more appropriate tool to represent and solve large-size optimization problems.

Additional investigation is required to implement the optimization problem with electricity sharing and verify that the computational advantage shown in this paper holds for the electricity sharing problem as well. Further, a comparison with other classical solvers that are currently used for real business applications (e.g., Gurobi or COIN-OR Couenne) would bring even more confidence. We reserve these and other investigations to future work. 

The goal of the paper was to show a current possibility of using quantum optimization in the energy industry and benchmark the results with existing classical approaches. Given the current geopolitical situation, and the related importance of energy sources at the moment, we hope to inspire researchers and companies to validate our findings and test this approach on other use cases in the energy sector.

\bibliography{bibtex}{}
\bibliographystyle{unsrt}

\appendix
\section{Approximate penalty function for inequality constraint}\label{sec:proofs}
The appendix focuses on the approximate penalty function introduced in \autoref{sec:ineq_penalty}.
The first section will show how this penalty function was derived.
The second section will provide a deeper analysis of this penalty function and its properties.

\subsection{Derivation on the approximate penalty function}
In this section we will focus on the derivation of the penalty function introduced in \autoref{sec:ineq_penalty}.
The goal is to construct a penalty function that favours values of $\bm{x}$ that comply with the constraint
\begin{equation}
\label{eq:app_constraint}
\bm{a}^\T\bm{x} \le b
\end{equation}
over values of $\bm{x}$ that violate the constraint.
Let $c$ be the tight lower bound of $\bm{a}^\T\bm{x}$.
Since $\bm{x}$ is a binary vector and $\bm{a}$ and $b$ are constants, such a $c$ always exists.
Furthermore, $c$ can be easily found.
If all values in $\bm{a}$ are non-negative, then the minimum of $\bm{a}^\T\bm{x}$ is when $\bm{x}=\bm{0}$ and therefore $c=0$.
In the case that $\bm{a}$ contains negative values, then $\bm{a}^\T\bm{x}$ is minimized when $\bm{x}$ is one for all elements corresponding with negative values in $\bm{a}$ and zero otherwise.
In this case we have that $c$ is the sum of all negative values.
Using this lower bound $c$ we can write \autoref{eq:app_constraint} as
\begin{equation}
    c \le \bm{a}^\T\bm{x} \le b.
\end{equation}
Next, we subtract $c$ from all sides to produce
\begin{equation}
\label{eq:intermediate_constraint}
0 \le \bm{a}^\T\bm{x} -c \le b-c.
\end{equation}
If $\bm{a}$, $b$ and $c$ are integer valued, one could replace the inequality constraint with an equality constraint by adding slack variables, which results in:
\begin{equation}
\label{eq:bad}
\bm{a}^\T\bm{x}-c = \sum_{i=0}^{I-2}2^iz_i + z_{I-1}(2^I-1 - b + c) \Rightarrow P_\text{int}(\bm{x},\bm{z}) = \left(\bm{a}^\T\bm{x}-c - \sum_{i=0}^{I-2}2^iz_i - z_{I-1}(2^I-1 - b + c)\right)^2.
\end{equation}
However, when $\bm{a}$, $b$ or $c$ are arbitrary non integer values, there can be $\bm{x}$ that comply with the constraint for which \autoref{eq:bad} does not hold for any $\bm{z}$.
For example, any $\bm{a^T\bm{x}=n-\frac{1}{2}}$, where $n<b$ is an integer, we have that the smallest value of $P_\text{int}$ is $\frac{1}{4}$.
Furthermore, when $\bm{a^T\bm{x}=b+\frac{1}{2}}$, which clearly violates the constraint, then there exists a $\bm{z}$, such that $P_\text{int}=\frac{1}{4}$.
Hence, $P_\text{int}$ cannot differentiate well between values that comply with the constraint and values that violate the constraint.
Therefore, if $\bm{a}$, $b$ or $c$ are arbitrary real valued constants, then \autoref{eq:bad} is not a good approximation of the constraint in \autoref{eq:app_constraint}.

We can overcome the problem by multiplying both sides by $\frac{2^K-1}{b-c}$, where $K$ is an integer valued constant, to produce
\begin{equation}
0 \le \frac{2^K-\frac{1}{2}}{b-c}\left(\bm{a}^\T\bm{x} - c\right) \le 2^K-\frac{1}{2}.
\end{equation}
The equation above can then be \emph{approximated} with the equality
\begin{equation}
    \frac{2^K-\frac{1}{2}}{b-c}\left(\bm{a}^\T\bm{x} - c\right) = \sum_{i=0}^{K-1}2^i z_i.
\end{equation}
Using the technique described in \autoref{sec:eq_penalty}, we can transform this equality constraint to the final penalty function:
\begin{equation}
\label{eq:app_final}
P(\bm{x},\bm{z}) = 
\left( \frac{2^K-\frac{1}{2}}{b-c}\left(\bm{a}^\T\bm{x} - c\right) - \sum_{i=0}^{K-1}2^i z_i\right)^2.
\end{equation}
In the next section we will prove that the penalty function in \autoref{eq:app_final} does indeed favour values of $\bm{x}$ that comply with the constraint over values that violate it.

\subsection{Analysis of the approximate penalty function}
In this section we will provide some deeper analysis on the approximate penalty function introduced in \autoref{sec:ineq_penalty}.
The goal of penalty function is to favour values of $x\in\{0,1\}^n$ that comply with
\begin{equation}
    \bm{a}^\T\bm{x} \le b,
\end{equation}
as opposed to values of $\bm{x}$ that violate the constraint above.
We start by showing that $\bm{x}$ meets the constraint, if and only if the there is a $\bm{z}$ such that $P(\bm{x},\bm{z})$ is between $0$ and $\frac{1}{4}$.
\begin{theorem}
\label{thm:penalty}
Let $(\bm{a}, b) \in \R^{n+1}$ defining the set $F = \{\bm{x}\in\{0,1\}^n \ | \ \bm{a}^\T\bm{x}\le b\}$.
Let $c=\min F$ be a lower bound of $\bm{a}^\T\bm{x}$ and
\[
P(\bm{x},\bm{z}) = 
\left( \frac{2^K-\frac{1}{2}}{b-c}\left(\bm{a}^\T\bm{x} - c\right) - \sum_{i=0}^{K-1}2^i z_i\right)^2,
\]
where $K\in \N$ and $\bm{z}\in\{0,1\}^{K}$.
Then, for $\bm{x}\in \{0,1\}^n$ the following statements are equivalent:
\begin{enumerate}
\item $\bm{x}\in F$.
\item There exists a $\bm{z}\in\{0,1\}^K$ such that $P(\bm{x},\bm{z}) \le \frac{1}{4}$.
\end{enumerate}
\end{theorem}

\begin{proof}
Since $\sum_{i=0}^{K-1}2^i z_i$ is a bijection from $\{0,1\}^K$ to $\{0,1,\ldots,2^K-1\}$, we shall replace $\sum_{i=0}^{K-1}2^i z_i$ by \mbox{$\hat{z}\in \{0,1,\ldots,2^K-1\}$} in this proof.
First we will proof $(1)\Rightarrow(2)$.
Suppose that $\bm{x} \in F$, then $\bm{a}^\T\bm{x}\le b$ and we have
\[
\frac{2^K-\frac{1}{2}}{b-c}\left(\bm{a}^\T\bm{x} - c\right)
\le \frac{2^K-\frac{1}{2}}{b-c}\left(b - c\right)
= 2^K-\frac{1}{2}.
\]
Because $c$ is a lower bound of $\bm{a}^\T\bm{x}$ we know
\[
\frac{2^K-\frac{1}{2}}{b-c}\left(\bm{a}^\T\bm{x} - c\right)
\ge \frac{2^K-\frac{1}{2}}{b-c}\left(c - c\right)
= 0.
\]
Define $\frac{2^K-\frac{1}{2}}{b-c}\left(\bm{a}^\T\bm{x} - c\right) = \hat{x}\in [0, 2^K-\frac{1}{2}]$, then
\[
\min_{\bm{z}\in\{0,1\}^K}P(\bm{x},\bm{z})
= \min_{\hat{z}\in \{0,1, \ldots, 2^K-1\}}( \hat{x} -\hat{z})^2
\le \frac{1}{4}.
\]
Hence, there is a $\bm{z}\in\{0,1\}^K$ such that $ P(\bm{x},\bm{z}) \le \frac{1}{4}$.
\\
\\
Next, we will complete the proof by showing $(2)\Rightarrow(1)$.
Let $\bm{x}\in \{0,1\}^n$ and suppose that we have $\bm{z}\in \{0,1\}^K$ such that $P(\bm{x}, \bm{z})\le \frac{1}{4}$.
Therefore,
\[
\frac{2^K-\frac{1}{2}}{b-c}\left(\bm{a}^\T\bm{x} - c\right) - \hat{z}\le \frac{1}{2}\]
Next we add $\hat{z}$ to both sides of the equation, multiply both sides by $\frac{b-c}{2^K-\frac{1}{2}}$ and finally add $c$ to both sides, which produces
\[\bm{a}^\T\bm{x} \le c + (\frac{1}{2}+\hat{z})\frac{b-c}{2^K-\frac{1}{2}}\]
Rewriting the right hand side gives
\[\frac{c(2^K-\frac{1}{2})+(\frac{1}{2}+\hat{z})(b-c)}{2^K-\frac{1}{2}}\]
Since $\hat{z} \in \{0, 1, \ldots, 2^K-1\}$ we know $\hat{z} \le 2^K-1$, hence
\[
\frac{c(2^K-\frac{1}{2})+(\frac{1}{2}+z)(b-c)}{2^K-\frac{1}{2}}
\le \frac{c(2^K-\frac{1}{2})+(2^K-\frac{1}{2})(b-c)}{2^K-\frac{1}{2}}
= b.
\]
Therefore, $\bm{a}^\T\bm{x} \le b$ and we are done.
\end{proof}

The consequence of \autoref{thm:penalty} is that $P(\bm{x},\bm{z})$ distinguishes between values of $\bm{x}$ that meet the constraint and that violate the constraint.
For a fixed $\bm{x}$ we have that if $\bm{x}$ meets the constraint, then the minimum of the penalty function is between $0$ and $\frac{1}{4}$.
On the other hand, if $\bm{x}$ violates the constraint, then the minimum of the penalty function is strictly larger then $\frac{1}{4}$. 

In the following proposition we shall show that this property is a direct consequence of the  constant $\frac{2^K-\frac{1}{2}}{b-c}$ in the definition of the penalty function.
\begin{proposition}\label{prop:unique}
Let $\bf{a}$, $b$, $c$ and $F$ be the same as defined in \autoref{thm:penalty}.
Suppose $P(\bm{x}, z) = \big(\alpha_K(\bm{a}^\T\bm{x}-c) - z\big)^2$, where \mbox{$z\in\{0,1,\ldots, 2^K-1\}$}.
Then the following statements are equivalent:
\begin{enumerate}
\item $\bm{x}\in F$.
\item There exists a $z$ such that $P(\bm{x},z) \le \frac{1}{4}$.
\end{enumerate}
if and only if $\alpha_K = \frac{2^K-\frac{1}{2}}{b-c}$.
\end{proposition}

\begin{proof}
The first direction of the proof is the case of \autoref{thm:penalty}.
For the other direction we divide the problem into four cases: (i) $\alpha_K > \frac{2^K-\frac{1}{2}}{b-c}$, (ii) $0 < \alpha_K < \frac{2^K-\frac{1}{2}}{b-c}$, (iii) $\alpha_K=0$ and (iv) $\alpha_K < 0$.
\\
\\
We start with the case (i) $\alpha_K > \frac{2^K-\frac{1}{2}}{b-c}$.
Suppose $\bm{x}^\T\bm{a}=b$, then $\bm{x}\in F$.
Therefore,
\[
\alpha(\bm{a}^\T\bm{x}-c) -z
> \frac{2^K-\frac{1}{2}}{b-c}(\bm{a}^\T\bm{x}-c) - z
= 2^K-\frac{1}{2} - z
\ge 2^K-\frac{1}{2} - (2^K-1) 
= \frac{1}{2}.
\]
Hence, for all $z\in \{0,1,\ldots,2^K-1\}$ we have $P(\bm{x},z)>\frac{1}{4}$.
\\
\\
For case (ii) $0 < \alpha_K < \frac{2^K-\frac{1}{2}}{b-c}$ we set $\bm{a}^T\bm{x}=b+\frac{1}{\alpha_K}$.
Clearly, $\bm{x}\notin F$, but we will show that there is a $z\in \{0,1,\ldots,2^K-1\}$ such that $P(\bm{x},z)\le \frac{1}{4}$.
By our choice of $\bm{x}$ we have that $\alpha(\bm{a}^\T\bm{x}-c) = \alpha(b-c) + 1$.
Furthermore, by the bounds on $\alpha$ we get $0<\alpha(b-c)<2^K-\frac{1}{2}$.
Combining both gives us
\[
0 < \alpha(\bm{a}^\T\bm{x}-c) < 2^K+\frac{1}{2}.
\]
Hence there exists a $z\in \{0,1,\ldots,2^K-1\}$ such that $P(\bm{x},z) < \frac{1}{4}$.
\\
\\
For case (iii) where $\alpha_K=0$ we take any $x\notin F$ and set $z=0$.
Then $P(\bm{x},z)=0$.
\\
\\
For case (iv) we have $\alpha_K < 0$.
Set $\bm{\alpha}^\T\bm{x} = c - \frac{1}{\alpha_K}$, where we assume that $b\ge c - \frac{1}{\alpha_K}$.
Therefore, $\bm{x}\in F$.
We then see
\[
\max_{z\in \{0,1,\ldots,2^K-1\}} \alpha_K(\bm{a}^T\bm{x}-c) -z
= \alpha_K(\bm{a}^T\bm{x}-c)
= -1
\]
Hence, for all $z\in \{0,1,\ldots,2^K-1\}$ we have $P(\bm{x}, z)\ge 1$.
\end{proof}
During the proof of \autoref{prop:unique} it becomes clear that $\alpha_K < \frac{2^K-\frac{1}{2}}{b-c}$ will not produce a useful penalty function as the penalty function does not distinguish between values of $\bm{x}$ that meet the constrain or ones that violate the constraint.
For the case $\alpha_K > \frac{2^K-\frac{1}{2}}{b-c}$, an argument can be made that it does differentiate between values of $\bm{x}$ that comply with the constraint and $\bm{x}$ that violate the constraint.
However, values of $\bm{x} \in F$, where $\bm{a}^T\bm{x}$ is close to $b$, will have much larger values for $\min_{\bm{z}}P(\bm{x},\bm{z})$, compared to values of $\bm{x}$, where $\bm{a}^\T\bm{x}$ is close to $c$.
Hence, a penalty function where $\alpha_K > \frac{2^K-\frac{1}{2}}{b-c}$ gives a worse approximation then a penalty function with $\alpha_K = \frac{2^K-\frac{1}{2}}{b-c}$.

It important to understand the relative size of the penalty that is given by the penalty function.
We will analyse this behaviour using the following proposition.
\begin{proposition}\label{prop:error_size}
Let $(\bm{a}, b) \in \R^{n+1}$ define the set $F = \{\bm{x}\in\{0,1\}^n \ | \ \bm{a}^\T\bm{x}\le b\}$.
Let $c = \min F$ and
\[
P(\bm{x},\bm{z}) = 
\left( \alpha_K(\bm{a}^\T\bm{x} - c) - \sum_{i=0}^{K-1}2^i z_i\right)^2,
\]
where $K\in \N$, $\bm{z}\in\{0,1\}^{K}$ and $\alpha_K=\frac{2^K-\frac{1}{2}}{b-c}$.
Then for $x\notin F$, 
\[
\min_{z\in\{0,1\}^K}P(\bm{x}, \bm{z}) = \frac{1}{4}+\alpha_k\varepsilon + (\alpha_k\varepsilon)^2,
\]
where $\varepsilon = \bm{a}^\T\bm{x}-b$.
\end{proposition}
\begin{proof}
Let $x\not\in F$, then $\bm{a}^\T\bm{x}>b$.
Therefore, we know $\alpha_K(\bm{a}^T\bm{x}-c) > 2^K-\frac{1}{2}$.
Hence, the $P(\bm{x},\bm{z})$ is minimized when all $z_i$ values are 1 and the total sum $\sum_{i=0}^{K-1}2^iz_i$ becomes $2^K-1$.
Thus, we get:
\[
\min_{z\in\{0,1\}^K}P(\bm{x}, \bm{z})
= (\alpha_K(\bm{a}^T\bm{x}-c)-2^K+1)^2.
\]
Next, we substitute $\bm{a}^\T\bm{x}=\varepsilon+b$ and use some basic algebra to produce
\[
\min_{z\in\{0,1\}^K}P(\bm{x}, \bm{z})
= (\alpha_K(b-c) + \alpha_K\varepsilon -2^K+1)^2.
\]
Note that $\alpha_K(b-c)=2^K-\frac{1}{2}$.
Using this property then produces the final result
\[
\min_{z\in\{0,1\}^K}P(\bm{x}, \bm{z})
= (2^K-\frac{1}{2} + \alpha_K\varepsilon -2^K+1)^2
= (\frac{1}{2} + \alpha_K\varepsilon)^2
= \frac{1}{4}+\alpha_k\varepsilon + (\alpha_k\varepsilon)^2
\]
\end{proof}
\autoref{prop:error_size} shows that the penalty given by $P(\bm{x},\bm{z})$ with $\alpha_K=\frac{2^K-\frac{1}{2}}{b-c}$ scales quadratically in both $\alpha_K$ and the amount of violation.
Because $\alpha_K$ scales exponentially in $K$, the sensitivity of the penalty function scales exponentially in $K$.
Hence, if small violations of the penalty occur, one can increase the sensitivity of the penalty function by incriminating $K$.
However, one must be wary to set $K$ to large values when implementing these penalties.
Since $\alpha_K$ scales exponentially in $K$, overflows are almost unavoidable when $K$ is large. 
\end{document}